\documentclass{cccg25}
\usepackage[T1]{fontenc}
\usepackage{amsmath}
\usepackage{amsfonts}
\usepackage{amssymb}
\usepackage{graphicx} 
\usepackage{adjustbox}
\usepackage{tikz}
\usetikzlibrary{calc}
\usepackage{subcaption}
\usepackage{thm-restate}

\usepackage{algorithm}
\usepackage{algorithmicx}
\usepackage{algpseudocode}

\algtext*{EndWhile}
\algtext*{EndFor}
\algtext*{EndIf}
\algtext*{EndFunction}

\newcommand{\largec}{l}
\newcommand{\mediumc}{b}
\newcommand{\smallc}{a}
\newcommand{\cost}[1]{C(#1)}
\title{Approximation and Hardness of Polychromatic TSP}
\author{
Thomas Schibler\thanks{University of California, Santa Barbara, CA, USA, \texttt{tschibler@ucsb.edu}}
\and
Subhash Suri\thanks{University of California Santa Barbara, CA, USA, \texttt{suri@ucsb.edu}}
\and
Jie Xue\thanks{New York University Shanghai, China, \texttt{jiexue@nyu.edu}}
}
\begin{document}
\thispagestyle{plain}
\maketitle

\begin{abstract}
We introduce the \emph{Polychromatic Traveling Salesman Problem} (PCTSP), where the input is an edge weighted graph whose vertices are partitioned into $k$ equal-sized color classes, and the goal is to find a minimum-length Hamiltonian cycle that visits the classes in a fixed cyclic order.  
This generalizes the Bipartite TSP (when $k = 2$) and the classical TSP (when $k = n$).  
We give a polynomial-time $(3 - 2\cdot10^{-36})$-approximation algorithm for metric PCTSP.  
Complementing this, we show that Euclidean PCTSP is APX-hard even in $\mathbb{R}^2$, ruling out the existence of a PTAS unless P = NP.
\end{abstract}

\section{Introduction}
The classic Traveling Salesman Problem (TSP) takes an edge-weighted graph $G$ as input and aims to find a minimum-weight Hamiltonian cycle in $G$ (which by definition is a simple cycle that visits all vertices of $G$).
As one of the most fundamental algorithmic problems, TSP has stood as a testing ground for algorithmic developments in the theory community since its introduction as one of Karp's original NP-Complete problems.
In 1976, Christofides showed a 1.5-approximation for any metric graph $G$, and this bound resisted progress for decades until Karlin et al. \cite{karlin2021} improved the bound (slightly) to $1.5 - 10^{-36}$, which remains the best known.
In the case that $G$ contains only edge weights 1 and 2, Papadimitriou and Yannakakis obtained a better approximation ratio of 7/6, and showed that even this restricted case is APX-Hard, and thus does not admit a PTAS unless P=NP \cite{papadimitriou1993}.
However, if the edge weights correspond to the Euclidean distances between points in some fixed $\mathbb{R}^d$, then a PTAS is known from the celebrated work of Arora~\cite{arora1998} and Mitchell.

In this paper, we consider a colored version of TSP, called \textit{polychromatic TSP} (PCTSP), in which the vertices of $G$ are partitioned into color classes of the same size and the goal is to find a TSP tour that repeatedly visits the color classes in a consistent order.
Formally, a \textit{polychromatic graph} consists of a graph $G$ together with a partition $\mathcal{P} = \{V_1, \cdots, V_k\}$ of $V(G)$ into subsets (called \textit{color classes}) of equal size (i.e., $|V_1| = \cdots = |V_k|$) and a weight function $w:E(G) \rightarrow \mathbb{R}_{\geq 0}$.
Let $(G,\mathcal{P},w)$ be a polychromatic graph with $|V(G)| = n$ and $|\mathcal{P}| = k$.
The \textit{weight} of a path/cycle $\pi$ in $G$, denoted by $w(\pi)$, is the sum of the weights of the edges in $\pi$ under the weight function $w$.
For a permutation $\sigma = (\sigma(0),\sigma(1),\dots,\sigma(k-1))$ of $[k] = \{1, \cdots, k\}$, we say a cycle $C$ in $G$ is a \textit{$\sigma$-cycle} if it can be written as $C = (v_0, v_1, \cdots, v_r)$ where $v_r = v_0$ such that $v_i \in V_{\sigma(i \text{ mod } k)}$ for all $i \in \{0\} \cup [r]$; the weight of $C$ is equal to $\sum_{i=1}^n w((v_{i-1},v_i))$.
Note that the length of a $\sigma$-cycle is always a multiple of $k$.
Intuitively, a $\sigma$-cycle repeatedly visits the color classes of $G$ in the order specified by $\sigma$.
The goal of PCTSP is to find a minimum-weight Hamiltonian cycle in $G$ that is a $\sigma$-cycle for some permutation $\sigma$ of $[k]$.
Clearly, when $k = n$, PCTSP is exactly the classic TSP.
As such, PCTSP is a generalization of TSP and is thus NP-hard even when $G$ is a metric graph.
The main focus of this paper is to study the approximability of PCTSP in metric graphs, and more specifically, Euclidean graphs in a fixed dimension.

We remark that PCTSP is a natural extension of Bipartite TSP, where a tour must alternate between red and blue nodes. This problem has long been studied in the robotics literature as an appropriate model for ``pick and place'' route planning \cite{anily1992, chalasani1996, srivastav2001}. (In this setting, a robot with unit capacity must move objects from a set of sources to a set of destinations, necessitating a route that alternates between the two node types.) 
For metric graphs, Anily and Hassin gave a $2.5$-approximation \cite{anily1992}, and Chalasani et al. \cite{chalasani1996} later improved the ratio to $2$.

Lastly, many other lines of research consider chromatic variants of the TSP.
For example, Dross et al. \cite{dross2023} give a Gap-ETH tight approximation scheme for the bicolored noncrossing Euclidean TSP, where the goal is to find a separate tour for the points of each color, such that the two do not intersect in $\mathbb{R}^2$. Baligács et al. give a constant factor approximation for the $3$-color version \cite{baligacs2024}.
Other settings involve multiple salesman that can only visit given subsets of the colors, or seek the shortest cycle that visits each color once.
In contrast, we aim for a single polychromatic tour that visits every node of $G$ while cycling through a large number of colors.



\subsection{Main results}

We initiate the study of metric PCTSP and Euclidean PCTSP in a fixed dimension.
Our first result is a polynomial-time constant-approximation algorithm for PCTSP in metric graphs.
\begin{restatable}{theorem}{approximation}
\label{thm:approximation}
    There is a polynomial-time $(3-2 \cdot 10^{-36})$-approximation algorithm for metric PCTSP.
\end{restatable}
To complement the above algorithmic result, we rule out the existence of (Q)PTASes for metric PCTSP and even Euclidean PCTSP, by proving the following APX-hardness result.
\begin{restatable}{theorem}{hardness}
\label{thm:hardness}
    Euclidean PCTSP in $\mathbb{R}^2$ does not admit a PTAS unless P=NP. In particular, Euclidean PCTSP is APX-Hard.
\end{restatable}


\paragraph{Organization.}
The rest of the paper is organized as follows.
In Section~\ref{sec-pre}, we give the basic notions and preliminaries that are required for our results.
Section~\ref{sec-approx} presents our approximation algorithm for metric PCTSP, and Section~\ref{sec-hardness} presents our hardness result for Euclidean PCTSP.

\section{Preliminaries} \label{sec-pre}
For a graph $G$, we use $V(G)$ and $E(G)$ to denote the vertex set and edge set of $G$, respectively.
For $V \subseteq V(G)$, we denote by $G[V]$ the graph induced in $G$ by $V$.
Let $(G,\mathcal{P},w)$ be a polychromatic graph where $\mathcal{P} = \{V_1,\dots,V_k\}$.
We call a cycle in $G$ a \textit{polychromatic cycle} in $(G,\mathcal{P},w)$ if it is a $\sigma$-cycle for some permutation of $[k]$.
We denote by $\mathcal{H}_G$ the set of all Hamiltonian $\sigma$-cycles of $G$: $\mathcal{H}_G = \{\pi: \text{$\pi$ is a polychromatic cycle and }V(\pi) = V(G)\}$.

In metric PCTSP, we are given a polychromatic graph $(G,\mathcal{P},w)$ such that $G$ is complete and $w((v, v')) \leq w((v, v'')) + w((v'', v'))$ for all $v, v', v'' \in V(G)$.
We are to compute $\arg\min_{\pi \in \mathcal{H}_G}{w(\pi)}$.
In Euclidean PCTSP, the input is $(P, \mathcal{P})$ for a set of points $P \subset \mathbb{R}^d$, and partition of the points $\mathcal{P} = \{P_1, \cdots, P_k\}$.
The task is to solve metric PCTSP on input $(G, \mathcal{P'}, w)$, where $G$ is the complete graph with vertices $V(G) = \{v_p : p \in P\}$, $\mathcal{P'} = \{P_1', \cdots, P_k'\}$ s.t. $P_i' =\{v_{p}:p \in P_i\}$ for all $1 \leq i \leq k$, and $w((v_{p}, v_{p'})) = ||p-p'||_2$. $||\cdot||_2$ denotes the Euclidean norm.
For convenience, we similarly refer to a cycle of the points $\pi = (p_0, \cdots, p_r)$ where $p_r = p_0$ as a $\sigma$-cycle if $p_i \in P_{\sigma(i \text{ mod } k)}$ for all $i \in \{0\} \cup [r]$.
The length of $\pi$ is $||\pi|| = \sum_{i\in[r]} ||p_i - p_{i+1}||_2$.
Finally, we use the term $\sigma$-tour to mean a Hamiltonian $\sigma$-cycle.
\section{Constant approximation for metric PCTSP}
\label{sec-approx}

In this section, we give our polynomial-time constant-approximation algorithm for metric PCTSP.
Consider an instance $(G,\mathcal{P},w)$ of metric PCTSP, where $|V(G)| = n$ and $\mathcal{P} = \{V_1,\dots,V_k\}$.

At a high level, we solve the problem via two steps.
First, we give an algorithm that can compute in polynomial time, for any given permutation $\sigma$ of $[k]$, an $O(1)$-approximation for the optimal $\sigma$-tour.
When $k=O(1)$, this already solves metric PCTSP since we can simply run the algorithm on all $k!$ possibilities of $\sigma$.
However, we do not have the assumption that $k$ is a constant in PCTSP.
Therefore, in the second step, we show how to approximate the optimal ordering $\sigma$ in polynomial time.

\subsection{Computing the tour for a fixed order}

We first consider how to compute a good $\sigma$-tour in $G$ for a given ordering $\sigma$ of $[k]$.
Without loss of generality, we can assume that $\sigma = (1,\dots,k)$.
For convenience, define $V_{k+1} = V_1$.
Our algorithm is presented in Algorithm~\ref{alg-fixedordertour}.
The sub-routine $\textsc{MinMatching}(G,w,V_i,V_{i+1})$ in line~2 computes a minimum-weight perfect matching $M_i$ between $V_i$ and $V_{i+1}$ under the weight function $w$, which must exist because $|V_i| = |V_{i+1}|$ and $G$ is a complete graph.
Define $H$ as the subgraph of $G$ consisting of the edges in $\bigcup_{i=1}^k M_i$ (line~3), and let $\mathcal{C}$ be the set of connected components of $H$ (line~4).
Note that each $C \in \mathcal{C}$ is a $\sigma$-cycle in $G$.
Next, we arbitrarily pick a vertex $v_C \in V(C) \cap V_1$ in each $C \in \mathcal{C}$ (line~6) and let $V_1' \subseteq V_1$ be the set of these vertices (line~7).
We then call the sub-routine $\textsc{ApproxTSP}(G[V_1'],w)$ in line~8, which computes a $(1.5-10^{-36})$-approximation TSP $T$ of the subgraph $G[V_1']$ under the weight function $w$, which can be done using the algorithm of Karlin et al.~\cite{karlin2021}.
Finally, we ``glue'' the cycles in $\mathcal{C}$ along $T$ as follows to obtain a Hamiltonian $\sigma$-cycle $S$ (line~9).
Write $\mathcal{C} = \{C_1,\dots,C_t\}$ such that $T = (v_{C_1},\dots,v_{C_k},v_{C_1})$.
For $i \in [t]$, let $u_{C_i}$ be the neighbor of $v_{C_i}$ in $C_t$ that belongs to $V_k$.
Then we take the disjoint $\sigma$-cycles $C_1,\dots,C_t$, remove the edges $(u_{C_1},v_{C_1}),\dots,(u_{C_k},v_{C_k})$, and add the edges $(u_{C_1},v_{C_2}),\dots,(u_{C_{k-1}},v_{C_k}), (u_{C_k},v_{C_1})$.
This results in a Hamiltonian cycle $S$ in $G$, and one can easily verify that $S$ is a $\sigma$-cycle.
We return $S$ as the output of our algorithm (line~10).

\begin{algorithm}[htbp]
    \caption{\textsc{FixedOrderTour}$(G,\{V_1, \cdots, V_k\},w)$}
    \begin{algorithmic}[1]
        \For{every $i \in [k]$}
            \State $M_i \leftarrow \textsc{MinMatching}(G,w,V_i,V_{i+1})$
        \EndFor
        \State $H \leftarrow (V(G),\bigcup_{i=1}^k M_i)$
        \State $\mathcal{C} \leftarrow$ set of connected components of $H$
        \For{every $C \in \mathcal{C}$}
            \State pick an arbitrary vertex $v_C \in V(C) \cap V_1$
        \EndFor
        \State $V_1' \leftarrow \{v_C: C \in \mathcal{C}\}$
        \State $T \leftarrow \textsc{ApproxTSP}(G[V_1'],w)$
        \State $S \leftarrow \textsc{Glue}(T,\mathcal{C})$
        \State \textbf{return} $S$
    \end{algorithmic}
    \label{alg-fixedordertour}
\end{algorithm}

In the rest of this section, we show that $S$ is a $(2.5-10^{-36})$-approximation for the optimal $\sigma$-tour of $G$.

\begin{lemma} \label{lem-boundingmatching}
    $\sum_{i=1}^k \sum_{e \in M_i} w(e) \leq \sum_{e \in E(R)} w(e)$ for any Hamiltonian $\sigma$-cycle $R$ in $G$.
\end{lemma}
\begin{proof}
Note that the edges of $R$ between $V_i$ and $V_{i+1}$ form a perfect matching between $V_i$ and $V_{i+1}$, for $i \in [k]$.
Since $M_i$ is a minimum-cost perfect matching between $V_i$ and $V_{i+1}$, we have $\sum_{i=1}^k \sum_{e \in M_i} w(e) \leq \sum_{e \in E(R)} w(e)$.
\end{proof}

\begin{lemma} \label{lem-match+tsp}
    $\sum_{e \in E(S)} w(e) \leq \sum_{i=1}^k \sum_{e \in M_i} w(e) + (1.5-10^{-36}) \cdot \mathsf{tsp}(G)$, where $\mathsf{tsp}(G)$ is the minimum weight of a Hamiltonian cycle of $G$.
\end{lemma}
\begin{proof}
For convenience, let us write $C_{t+1} = C_1$.
We observe that $\sum_{e \in E(S)} w(e) = \sum_{i=1}^k \sum_{e \in M_i} w(e) + \sum_{i=1}^t (w((u_{C_i},v_{C_{i+1}}))-w((u_{C_i},v_{C_i})))$.
Since $G$ is a metric graph, we have $w((u_{C_i},v_{C_{i+1}}))-w((u_{C_i},v_{C_i})) \leq w((v_{C_i},v_{C_{i+1}}))$.
It follows that $\sum_{i=1}^t (w((u_{C_i},v_{C_{i+1}}))-w((u_{C_i},v_{C_i}))) \leq \sum_{e \in E(T)} w(e)$.
Recall that $T$ is a $(1.5-10^{-36})$-approximation TSP $T$ of $G[V_1']$.
As $G[V_1']$ is an induced subgraph of $G$ and $G$ is a metric graph, we have $\sum_{e \in E(T)} w(e) \leq (1.5-10^{-36}) \mathsf{tsp}(G)$.
\end{proof}

Since $\mathsf{tsp}(G) \leq \sum_{e \in E(R)} w(e)$ for any Hamiltonian $\sigma$-cycle $R$ in $G$, the above lemmas imply that $\sum_{e \in E(S)} w(e) \leq (2.5-10^{-36}) \sum_{e \in E(R)} w(e)$ for any Hamiltonian $\sigma$-cycle $R$ in $G$.

\subsection{Solving metric PCTSP}

In this section, we show how to compute a good ordering $\sigma$.
Together with the algorithm in the previous section, this shall give us the final algorithm for Theorem~\ref{thm:approximation}.

For $i,j \in [k]$ with $i \neq j$, we compute a minimum-cost perfect matching $M_{i,j}$ between $V_i$ and $V_j$ in $G$.
Let $\Sigma_k$ denote the set of all permutations of $[k]$.
For $\sigma = (r_1,\dots,r_k) \in \Sigma_k$, let $H_\sigma$ denote the subgraph of $G$ consisting of the edges in $\bigcup_{i=1}^k M_{r_i,r_{i+1}}$; here we set $r_{k+1} = r_1$ for convenience.
Define $w_\sigma = \sum_{e \in E(H_\sigma)} w(e)$.

Suppose $\mathsf{opt}$ is the optimum of the PCTSP instance, i.e., the weight of an optimal solution.
Note that $\mathsf{opt} \geq \mathsf{tsp}(G)$.
By Lemma~\ref{lem-boundingmatching}, we have $\mathsf{opt} \geq \min_{\sigma \in \Sigma_k} w_\sigma$.
If $\hat{\sigma} \in \Sigma_k$ satisfies that $w_{\hat{\sigma}} \leq c \cdot \min_{\sigma \in \Sigma_k} w_\sigma$, then Lemma~\ref{lem-match+tsp} further implies that when running the algorithm in the previous section on $\hat{\sigma}$, the output tour $S$ satisfies $\sum_{e \in E(S)} w(e) \leq (c+1.5-10^{-36}) \cdot \mathsf{opt}$.
Based on this observation, it suffices to find $\hat{\sigma} \in \Sigma_k$ that (approximately) minimizes $w_{\hat{\sigma}}$.

We achieve this goal as follows.
Build a complete graph $G'$ with $V(G') = \{V_1,\dots,V_k\}$, and define a weight function $w':E(G') \rightarrow \mathbb{R}_{\geq 0}$ by setting $w'((V_i,V_j)) = \sum_{e \in E(M_{i,j})} w(e)$.
One can easily check that $G'$ is also a metric graph.
Observe that for each $\sigma = (r_1,\dots,r_k) \in \Sigma_k$, $w_\sigma$ is just equal to the weight of the Hamiltonian cycle $(V_{r_1},\dots,V_{r_k},V_{r_1})$ in $G'$ under $w'$.
As such, we simply compute a $(1.5-10^{-36})$-approximation TSP $(V_{\hat{r}_1},\dots,V_{\hat{r}_k},V_{\hat{r}_1})$ in $G'$.
Then $\hat{\sigma} = (\hat{r}_1,\dots,\hat{r}_k)$ satisfies that $w_{\hat{\sigma}} \leq (1.5-10^{-36}) \cdot \min_{\sigma \in \Sigma_k} w_\sigma$.
Therefore, if we apply the algorithm in the previous section on $\hat{\sigma}$, the output tour $S$ satisfies $\sum_{e \in E(S)} w(e) \leq (3-10^{-36}) \cdot \mathsf{opt}$.
This completes the proof of Theorem~\ref{thm:approximation}.


\approximation*


\paragraph{Remark.}
As one can verify from our algorithm, the approximation factor and the running time of the algorithm in Theorem~\ref{thm:approximation} in fact depend on the best matching algorithm and metric TSP algorithm. 
Specifically, let $M(n)$ be the time for computing a minimum-cost bipartite matching in a complete bipartite graph with $n$ vertices and $T(n)$ be the time for computing an $\alpha$-approximate TSP in a metric graph with $n$ vertices.
Then we get a $2\alpha$-approximation algorithm for PCTSP with running time $O(k^2 M(2n/k) + T(k) + T(n/k))$.

\newcommand{\seta}{A}
\newcommand{\setb}{B}
\newcommand{\setc}{C}

\section{APX-Hardness for Euclidean PCTSP in $\mathbb{R}^2$}
\label{sec-hardness}

In this section, we prove Theorem \ref{thm:hardness} by a reduction from Max 2-SAT, which is known to be APX-hard \cite{papadimitriou1991}.
The input to Max 2-SAT is a set of $n$ variables $\{x_1, \cdots, x_n\}$ and $m$ clauses $\{c_1, \cdots, c_m\}$. Each clause is a conjunction of two literals; a literal is a variable $x_i$ or its negation $\bar{x}_i$.
The task is to find a boolean assignment of the variables that satisfies the maximum number of clauses.

The general idea behind our reduction is to encode the truth assignment of SAT variables in the choice of permutation $\sigma$. We then test a corresponding tour with a series of clause gadgets; informally, if $\sigma$ fails to ``satisfy'' a clause, then the clause gadget will penalize the tour by a fixed small amount. The structure of the proof is as follows. We first describe how to encode a truth assignment using $\sigma$. We then introduce our clause gadget and reason about its structure in Lemma ~\ref{lem:gadget}. Finally, we build the full construction using a series of clause gadgets and an auxiliary set of points $S$. The points of $S$ serve two purposes: (1) they rule out the possibility of $\sigma$ that does not properly encode a truth assignment, and (2) they divert the return trip of the tour away from the clause gadgets so that we may reason about a single pass through our gadgets. 

Our construction uses $k = 3n + 1$ color classes (point sets), which we will label $R_\alpha$, $T_\alpha$, $F_\alpha$ for $\alpha \in [1, n]$, and the additional $R_{n+1}$.
For a permutation $\sigma$, let $\prec_\sigma$ be the total order of the classes under $\sigma$, i.e. $V_{\sigma(0)} \prec_\sigma V_{\sigma(1)} \cdots \prec_\sigma V_{\sigma(k-1)}$.
We say that a permutation $\sigma$ is \textit{valid} if it satisfies:
\begin{enumerate}
\item $R_\alpha \prec_\sigma T_\alpha$, $R_\alpha \prec_\sigma F_\alpha$, for all $1 \leq \alpha \leq n$ and
\item $T_\alpha \prec_\sigma R_{\beta}$, $F_\alpha \prec_\sigma R_{\beta}$, for all $1 \leq \alpha < \beta \leq n+1$.
\end{enumerate}
Observe that there are exactly $2^n$ valid permutations that naturally correspond to truth assignments of the variables $\{x_1, \cdots, x_n\}$. 
For each $\alpha \in [1, n]$, we only have the freedom to choose whether $T_\alpha \prec_\sigma F_\alpha$ or $F_\alpha \prec_\sigma T_\alpha$.
We say that $\sigma$ \emph{satisfies} the literal $x_\alpha$ if $T_\alpha \prec_\sigma F_\alpha$; likewise $\sigma$ satisfies $\bar{x}_\alpha$ if $F_\alpha \prec_\sigma T_\alpha$.
$\sigma$ satisfies a clause $c_i$ if it satisfies either of the literals of $c_i$.

\subsection{Clause gadget}

The crux of our reduction lies in the creation of a clause gadget that tests $\sigma$ on whether it satisfies a given clause.
This test can only depend on $\sigma$'s encoding of the two variables in that clause, so the gadget will need to allow the tour to ``jump through'' all of the colors that precede $\{T_\alpha, F_\alpha\}$, that fall between $\{T_\alpha, F_\alpha\}$, $\{T_\beta, F_\beta\}$, and that come after $\{T_\beta, F_\beta\}$, respectively. 
Formally, we define these color sets for use in the gadget, for all $\alpha, \beta$, s.t. $1 \leq \alpha \leq \beta \leq n$:
\begin{itemize}
\item $\seta_{\alpha,\beta} = \{R_{\alpha'}, T_{\alpha'}, F_{\alpha'} | 0 < \alpha' < \alpha\} \cup \{R_\alpha\}$
\item $\setb_{\alpha,\beta} = \{R_{\alpha'}, T_{\alpha'}, F_{\alpha'} | \alpha < \alpha' < \beta\} \cup \{R_\beta\}$
\item $\setc_{\alpha,\beta} = \{R_{\alpha'}, T_{\alpha'}, F_{\alpha'} | \beta < \alpha' < n+1\} \cup \{R_{n+1}\}$
\end{itemize}
When the indices $\alpha$ and $\beta$ are clear, we abbreviate the sets as $\seta, \setb, \setc$.
These sets are nonempty, and any valid $\sigma$ satisfies (by definition):
\begin{equation}
\label{eq:valid}
\seta \prec_\sigma \{T_\alpha, F_\alpha\}, \prec_\sigma \setb \prec_\sigma \{T_\beta, F_\beta\} \prec_\sigma \setc
\end{equation}

Figure \ref{fig:gadget} shows the gadget, along with (optimal) example paths for when $\sigma$ does, and does not satisfy the clause.
We describe the position of the points of the gadget (and later the spacing between gadgets) in terms of constants $\smallc$ and $\mediumc$, with the intuition that $\mediumc \gg \smallc \gg 1$.
Let $x > 0$.
A clause gadget $C_i$ for the clause $c_i = x_\alpha \vee x_\beta$ anchored at $x$ is the following collection of points.\footnote{For clauses that include negative literals (e.g. $\bar{x}_\alpha$ instead of $x_\alpha$), we swap the corresponding colors ($T_\alpha$ and $F_\alpha$) in the gadget construction.
The key is that the satisfying color appears on the line $y = -1$, and the other on the line $y = 1$.
Without loss of generality, we will assume these colors are labeled as $T_\alpha$ and $T_\beta$ in our analysis.}
(When we refer to a set of colors $X$ at a location, we mean a set of $|X|$ coincident points at that location, one of each color in the set $X$.)
\begin{itemize}
    \item Set $\seta$ at $(x, -1)$
    \item Points of color $F_\alpha$, $T_\alpha$ at $(x + \smallc, 1), (x + \smallc, -1)$, resp.
    \item Two sets $\setb$ at $(x+2\smallc, \pm1)$
    \item Points of $F_\beta$, $T_\beta$ at $(x + 2\smallc, 1), (x + 2\smallc, -1)$, resp.
    \item Set $\setc$ at $(x + 3\smallc, 1)$
    \item Set $\seta \cup \{T_\alpha,F_\alpha\}$ at $(x + 3\smallc, 0)$
    \item Set $\{T_\alpha, F_\alpha\}$ at $(x + 9\smallc, 0)$ 
    \item Set $\setc$ at $(x + 27\smallc, -1)$ 
\end{itemize}

\begin{figure*}
    \centering
    \newcommand{\drawGraph}[2]{
        \begin{subfigure}{\textwidth} 
            \centering
            \begin{tikzpicture}
                \def\a{2.5}
                \def\dx{0.5} 
                \def\stretch{1.2} 
                \def\stretchh{0.5} 
                \def\h{0.8}
                \def\arrowHeight{\h+0.65} 
                \def\sep{2}

                \node (p1) at (0, -\h) [circle, draw, inner sep=\sep pt] {\small $p_1$};
                \node (p2) at (\a, \h) [circle, draw, inner sep=\sep pt] {\small $p_2$};
                \node (p3) at (\a, -\h) [circle, draw, inner sep=\sep pt] {\small $p_3$};
                \node (p4) at (2*\a-\dx, \h) [circle, draw, inner sep=\sep pt] {\small $p_4$};
                \node (p5) at (2*\a-\dx, -\h) [circle, draw, inner sep=\sep pt] {\small $p_5$};
                \node (p6) at (2*\a+\dx, \h) [circle, draw, inner sep=\sep pt] {\small $p_6$};
                \node (p7) at (2*\a+\dx, -\h) [circle, draw, inner sep=\sep pt] {\small $p_7$};
                \node (p8) at (3*\a, \h) [circle, draw, inner sep=\sep pt] {\small $p_8$};
                \node (p9) at (3*\a, 0) [circle, draw, inner sep=\sep pt] {\small $p_9$};
                \node (p10) at (3*\a+\stretch*\a, 0) [circle, draw, inner sep=\sep pt] {\small $p_{10}$};
                \node (p11) at (3*\a+\stretch*\a+\stretchh*\a, -\h) [circle, draw, inner sep=\sep pt] {\small $p_{11}$};

                \node [above left, xshift=-4pt, yshift=1pt] at (p1) {\small $\seta$};
                \node [above left, xshift=-4pt, yshift=1pt] at (p2) {\small $F_\alpha$};  
                \node [below left, xshift=-4pt, yshift=-1pt] at (p3) {\small $T_\alpha$};  
                \node [above left, xshift=-4pt, yshift=1pt] at (p4) {\small $\setb$};  
                \node [below left, xshift=-4pt, yshift=-1pt] at (p5) {\small $\setb$};  
                \node [above right, xshift=4pt, yshift=0pt] at (p6) {\small $F_\beta$};
                \node [below right, xshift=4pt, yshift=-1pt] at (p7) {\small $T_\beta$};
                \node [above right, xshift=4pt, yshift=1pt] at (p8) {\small $\setc$};
                \node [right, xshift=3pt] at (p9) {\small $\seta \cup \{T_\alpha, F_\alpha\}$};
                \node [above right, xshift=4pt, yshift=1pt] at (p10) {\small $\{T_\beta, F_\beta\}$};
                \node [above right, xshift=4pt, yshift=2pt] at (p11) {\small $\setc$};

                \draw[<->] (0, \arrowHeight) -- (\a, \arrowHeight) node[midway, above] {\small $a$};
                \draw[<->] (\a, \arrowHeight) -- (2*\a-\dx, \arrowHeight) node[midway, above] {\small $a$};
                \draw[<->] (2*\a-\dx, \arrowHeight) -- (2*\a+\dx, \arrowHeight) node[midway, above] {\small $0$};  
                \draw[<->] (2*\a+\dx, \arrowHeight) -- (3*\a, \arrowHeight) node[midway, above] {\small $a$};
                \draw[<->] (3*\a, \arrowHeight) -- (3*\a+\stretch*\a, \arrowHeight) node[midway, above] {\small $6a$};
                \draw[<->, dashed] (3*\a+\stretch*\a, \arrowHeight) -- (3*\a+\stretch*\a+\stretchh*\a, \arrowHeight) node[midway, above] {\small $18a$};
                
                #2

            \end{tikzpicture}
            \caption{#1} 
        \end{subfigure}
    }

    \drawGraph{Optimal $\sigma$-path visiting $C_i$ when $\sigma$ does not satisfy the clause, i.e. $F_\alpha \prec_\sigma T_\alpha$, $F_\beta \prec_\sigma T_\beta$}{ 
        \draw[->, thick] (p1) -- (p2);
        \draw[->, thick] (p2) -- (p3);
        \draw[->, thick] (p3) -- (p4);
        \draw[->, thick] (p4) -- (p6);
        \draw[->, thick] (p6) -- (p7);
        \draw[->, thick] (p7) -- (p8);
        \draw[->, thick] (p8) -- (p9);
        \draw[->, thick] (p9) -- (p5);
        \draw[->, thick] (p5) -- (p10);
        \draw[->, thick, dashed] (p10) -- (p11);
    }

    \drawGraph{Optimal $\sigma$-path visiting $C_i$ where $\sigma$ satisfies $x_\alpha$ but not $x_\beta$, i.e. $T_\alpha \prec_\sigma F_\alpha$, $F_\beta \prec_\sigma T_\beta$}{ 
        \draw[->, thick] (p1) -- (p3);
        \draw[->, thick] (p3) -- (p2);
        \draw[->, thick] (p2) -- (p4);
        \draw[->, thick] (p4) -- (p6);
        \draw[->, thick] (p6) -- (p7);
        \draw[->, thick] (p7) -- (p8);
        \draw[->, thick] (p8) -- (p9);
        \draw[->, thick] (p9) -- (p5);
        \draw[->, thick] (p5) -- (p10);
        \draw[->, thick, dashed] (p10) -- (p11);
    }

    \drawGraph{Optimal $\sigma$-path visiting $C_i$ where $\sigma$ satisfies $x_\beta$ but not $x_\alpha$, i.e. $F_\alpha \prec_\sigma T_\alpha$, $T_\beta \prec_\sigma F_\beta$}{ 
        \draw[->, thick] (p1) -- (p2);
        \draw[->, thick] (p2) -- (p3);
        \draw[->, thick] (p3) -- (p5);
        \draw[->, thick] (p5) -- (p7);
        \draw[->, thick] (p7) -- (p6);
        \draw[->, thick] (p6) -- (p8);
        \draw[->, thick] (p8) -- (p9);
        \draw[->, thick] (p9) -- (p4);
        \draw[->, thick] (p4) -- (p10);
        \draw[->, thick, dashed] (p10) -- (p11);
    }

    \drawGraph{Optimal $\sigma$-path visiting $C_i$ where $\sigma$ satisfies both $x_\alpha$ and $x_\beta$, i.e. $T_\alpha \prec_\sigma F_\alpha$, $T_\beta \prec_\sigma F_\beta$}{ 
        \draw[->, thick] (p1) -- (p3);
        \draw[->, thick] (p3) -- (p2);
        \draw[->, thick] (p2) -- (p5);
        \draw[->, thick] (p5) -- (p7);
        \draw[->, thick] (p7) -- (p6);
        \draw[->, thick] (p6) -- (p8);
        \draw[->, thick] (p8) -- (p9);
        \draw[->, thick] (p9) -- (p4);
        \draw[->, thick] (p4) -- (p10);
        \draw[->, thick, dashed] (p10) -- (p11);
    }
    
    \caption{The clause gadget $C_i$ for clause $c_i = x_\alpha \lor x_\beta$ showing four possible $\sigma$-paths.
    The gadget consists of $6m + 2$ points (two of each color) placed at 9 distinct locations.
    (Notice that $p_4, p_6$ and $p_5, p_7$ coincide, but we distinguish them to make it clear when a path revisits a location.)
    At an intuitive level, the gadget tests whether $\sigma$ satisfies $x_\alpha$ (i.e. that $T_\alpha \prec_\sigma F_\alpha$) in the triangle $(p_1, p_2, p_3)$.
    Observe that path $(a)$ does not satisfy $x_\alpha$, and must take the longer leg $(p_1, p_2)$, while path $(b)$ takes the shorter leg $(p_2, p_3)$ by satisfying $x_\alpha$.
    Similarly, triangle $(p_6, p_7, p_8)$ tests $x_\beta$. 
    However, between these two tests, the tour has the choice of which set $B$ to visit, either at $p_4$ or $p_5$, such that satisfying exactly one of $x_\alpha$ or $x_\beta$ allows the tour to take an edge of rectangle $(p_2, p_3, p_5, p_4)$ while satisfying both (or neither) of $x_\alpha$, $x_\beta$ forces the path to take a diagonal.
    This avoids over-rewarding $\sigma$ for satisfying both $x_\alpha$ and $x_\beta$, and is the main idea behind the gadget.}
    \label{fig:gadget}
\end{figure*}

Let us first establish the length of the optimal path through the gadget.
To do so, we need to clarify what it means for a path to be polychromatic.
A path is a \emph{$\sigma$-path} if it can be written as $(v_1, \cdots, v_{r-1})$ such that for some $j$, $v_i \in V_{\sigma(i+j \mod k)}$ for all $i \in [r]$.
Note that a $\sigma$-path may start at a point of any color.
We will also need the following definitions.
Let a permutation $\sigma$ be \emph{$\alpha,\beta$-valid} if equation \ref{eq:valid} holds for $\sigma, \alpha, \beta$.
Two permutations $\sigma$, $\sigma'$ \emph{agree} on $x_\alpha$ if $T_\alpha \prec_\sigma F_\alpha$ iff $T_\alpha \prec_{\sigma'} F_\alpha$.

\begin{figure*}
    \centering
    \adjustbox{max width=\textwidth} {
    \begin{tikzpicture}
        \def\W{16}  
        \def\h{-2.3}   
        \def\a{2}   
        \def\b{1.5} 
        \def\m{5}   
        
        \def\shiftX{0} 

        \draw[thick] (\shiftX, 0) rectangle (\shiftX + \a, -1);
        \node at (\shiftX + \a/2, -0.5) {\small $C_1$};

        \draw[thick] (\shiftX + \a + \b, 0) rectangle (\shiftX + \a + \b + \a, -1);
        \node at (\shiftX + \a + \b + \a/2, -0.5) {\small $C_2$};

        \draw[thick, dotted] (\shiftX + \a + \b + \a + \b/2, -0.5) -- (\shiftX + \W - \a - \b/2, -0.5);

        \draw[thick] (\shiftX + \W - \a, 0) rectangle (\shiftX + \W, -1);
        \node at (\shiftX + \W - \a/2, -0.5) {\small $C_m$};

        \draw[<->] (\shiftX, -1.25) -- (\shiftX + \a, -1.25);
        \node at (\shiftX + \a/2, -1.4) {\small $27a$};

        \draw[<->] (\shiftX + \a, -1.25) -- (\shiftX + \a + \b, -1.25);
        \node at (\shiftX + \a + \b/2, -1.4) {\small $b$};

        \draw[<->] (\shiftX, 0.25) -- (\shiftX + \W, 0.25);
        \node at (\shiftX + \W/2, 0.5) {\small $W = 27ma + (m-1)b$};

        \def\xgap{4}  
        \def\ygap{0.2} 
        \def\dotsep{1} 

        \node (Ri+1) at (\shiftX + \xgap,\h) [circle,fill=black,inner sep=2pt,label=below:$R_{i+1}$] {};
        \node (Ti) at (\shiftX + \W/2,\h+\ygap) [circle,fill=black,inner sep=2pt,label=above:$T_i$] {}; 
        \node (Fi) at (\shiftX + \W/2,\h-\ygap) [circle,fill=black,inner sep=2pt,label=below:$F_i$] {}; 
        \node (Ri) at (\shiftX + \W - \xgap,\h) [circle,fill=black,inner sep=2pt,label=below:$R_i$] {};

        \draw[<->] (Ri+1) -- node[below] {$l n_i$} ($(Fi) + (0, \ygap)$);
        \draw[<->] ($(Fi) + (0, \ygap)$) -- node[below] {$l n_i$} (Ri);

        \node (Rn+1) at (\shiftX, \h) [circle,fill=black,inner sep=2pt,label=below:$R_{n+1}$] {};
        \node (R1) at (\shiftX + \W, \h) [circle,fill=black,inner sep=2pt,label=below:$R_1$] {};

        \draw[dotted, thick] (R1) -- (Ri);
        \draw[dotted, thick] (Ri+1) -- (Rn+1);

        \draw[<->] (5, -\ygap-1) -- (5, \h+\ygap);
        \node at (5.4, \h/2-0.5) {\small $2W$};

    \end{tikzpicture}
    }
    \caption{The instance $P_\Pi = S \cup C$. $C$ (top) is the union of $m$ clause gadgets $C_1 \cdots C_m$. $S$ (bottom) contains exactly one point of each color placed along the line $y = -2W - 1$. }
    \label{fig:reduction}
\end{figure*}
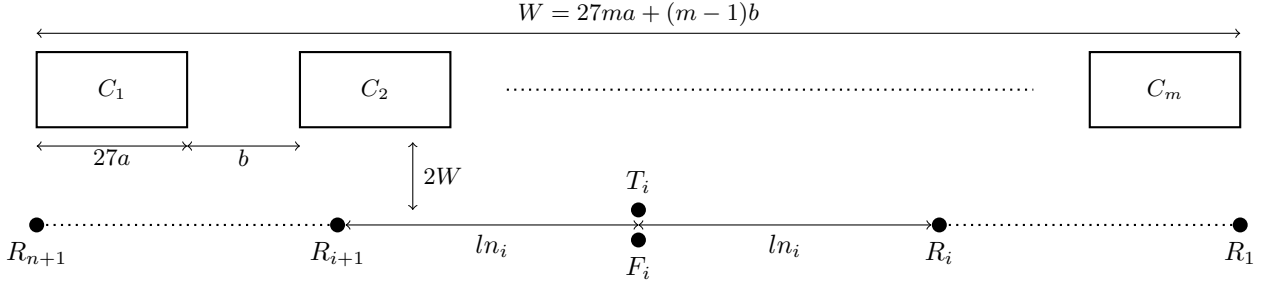

\begin{lemma}
\label{lem:gadget}
Let $C_i$ be the clause gadget for $c_i = x_\alpha \vee x_\beta$. For any valid $\sigma$, let $T_i^*$ be the shortest $\sigma$-path that visits all points of $C_i$.
Then,
\begin{equation*}
\|T_i^*\| = 
\begin{cases}
c + 2\smallc & \text{if $\sigma$ satisfies $c_i$}\\
c + 2\sqrt{\smallc^2 + 4} & \text{otherwise}
\end{cases}
\end{equation*}
where $c = 5 + \sqrt{\smallc^2 + 4} + \sqrt{\smallc^2 + 1} + \sqrt{49\smallc^2 + 1} + \sqrt{(18\smallc)^2 + 1}$. Furthermore, there is no $\sigma'$ and $\sigma'$-path shorter than $T_i^*$ visiting $C_i$ such that $\sigma'$ is $\alpha,\beta$-valid and agrees with $\sigma$ on $x_\alpha, x_\beta$.
\end{lemma}
\begin{proof}
We first give a constructive proof of the existence of a path with the desired cost.
Of course, we cannot hope to give a path for each of the $2^n$ possible valid orderings, so we will need to argue that one can be easily constructed for any $\sigma$ from one of a few candidate paths.
We give four such candidate paths, one for each pair of choices $T_\alpha \prec_\sigma F_\alpha$ or $F_\alpha \prec_\sigma T_\alpha$ and $T_\beta \prec_\sigma F_\beta$ or $F_\beta \prec_\sigma T_\beta$; see Figure \ref{fig:gadget}.
The paths are depicted such that upon visiting a set of points at any location $p_i$, each path visits all points of $c_i$ located there in the order specified by $\sigma$. 
(We specifically refer to $p_1, \cdots, p_{11}$ as \emph{locations}, reserving the term \emph{point} for $q \in C_i$ to avoid confusion.)
Path $(a)$ is the unique path for which $\sigma$ does not satisfy the clause $c_i$, namely, $F_\alpha \prec_\sigma T_\alpha$ and $F_\beta \prec_\sigma T_\beta$.
We can manually verify that this path is strictly longer than the other three, each of which have the same desired length.
The purpose of the four paths is that for any $\sigma'$, $\sigma'$ agrees with $\sigma$ on both $x_\alpha$ and $x_\beta$ for one of the four candidate $\sigma$.
This makes it easy to convert the $\sigma$-path to a $\sigma'$-path with the same cost.
For $j \neq \alpha, \beta$, each point of color $T_j$ in $C_i$ is coincident to a point of color $F_j$ and vice versa.
Since each path of our four candidate $\sigma$-paths visits these corresponding points one immediately after the other, then for each variable $x_j$ in which $\sigma$ and $\sigma'$ do not agree, upon visiting a point of $T_j$ (resp. $F_j$), the path can instead first visit the corresponding point of color $F_j$ (resp. $T_j$).
Applying this swap for all appropriate $j$ converts a $\sigma$-path to a $\sigma'$-path without changing its weight.

We now prove the nonexistence of a shorter $\sigma'$-path for any $\sigma'$ that is $\alpha,\beta$-valid.
First, observe that if we project each of the four tours onto the $x$-axis, then the projected tours each have length exactly $29a$.
Assuming that constant $a$ is chosen so that the total vertical deviation amounts to less than $a$ (e.g. $a > 14$), then each of the original tours have length less than $30a$.
For (almost) any other $\sigma'$-path, we will show that its length after projecting onto the $x$-axis remains at least $30a$.

Specifically, let $T^{\sigma'} = (q_1, \cdots, q_{6m+2})$ be a $\sigma'$-path visiting $C_i$ for some $\alpha,\beta$-valid $\sigma'$ (and thus $C_i = \{q_1, \cdots, q_{6m+2}\}$ an appropriate labeling of the points of the gadget).
By definition, $\sigma'$ satisfies equation \ref{eq:valid}, namely,
\[\seta \prec_{\sigma'} \{T_\alpha, F_\beta\}, \prec_{\sigma'} \setb \prec_{\sigma'} \{T_\beta, F_\beta\} \prec_{\sigma'} \setc.\]
We will also use the fact that the gadget contains exactly two points of each color, which limits the possible path structures.
Specifically, we fix a bijection of the points $M: C_i \rightarrow C_i$, where $M(q_j)$ is the unique point of $C_i$ with the same color as $q_j$.
For any $q_j$, a path must visit a point of every other color between visiting the points $q_j$ and $M(q_j)$.
Otherwise, the path is not a $\sigma'$-path for any $\sigma'$.

We start by showing that $T^{\sigma'}$ last visits all of the points at location $p_{11}$.
For sake of contradiction, assume otherwise, and that $||T^{\sigma'}|| < 30a$.
Then $\pi$ goes to and from $q_{11}$ implying $||\pi|| > 36a$, or else $\pi$ first visits $p_{11}$ and never returns.
In the latter case, after visiting all points of $C$ at $p_{11}$, $T^{\sigma'}$ must visit a point in $A$ by equation \ref{eq:valid}, either at location $p_1$ or $p_9$.
In either case, this leg has length at least $24a$, and eventually $T^{\sigma'}$ goes to $p_{10}$ with a leg of length at least $6a$, which sums to $30a$, which contradicts our assumption. So, $T^{\sigma'}$ last visits all of the points at location $p_{11}$.

This implies that $T^{\sigma'}$ starts at a point of $A$ (by equation \ref{eq:valid}), again either at $p_1$ or $p_9$.
If $T^{\sigma'}$ starts at $p_9$ (or visits a point of $p_9$ before visiting all points at $p_1$), it later visits $p_1$, and eventually $p_{11}$. 
The leg from $p_9$ to $p_1$ has length least $3a$ and the leg from $p_1$ to $p_{11}$ has length at least $27a$, so the total is again at least $30a$.
Therefore, $T^{\sigma'}$ first visits all points at location $p_1$, and last visits all points at location $p_{11}$.

Now, notice that each of the tours of Figure \ref{fig:gadget} travel almost monotonically in the horizontal ($x$) dimension from $p_1$ to $p_{11}$, only backtracking a distance $a$ from $p_9$ to either $p_4$ or $p_5$. 
On the other hand, any path that backtracks at least a distance $2a$ while traveling from $p_1$ to $p_{11}$ has a total length at least $30a$.
Notice that locations $p_4$ and $p_5$ share an $x$-coordinate, and collectively coincide with all points of color $\setb$.
Thus, any polychromatic path must first visit a point in one of them, then visit points of all remaining colors, and then return to the other one.
To do so necessarily means that the path backtracks a distance $a$, so this is the only allowable backtracking to remain under cost $30a$.
We use this observation to eliminate almost all of the remaining paths.
In particular, every pair of points $q_j$, $M(q_j)$ of matching color (excluding those located at $p_4$ or $p_5$) are separated by at least a horizontal distance $a$, so a path of length $< 30a$ must visit the one with lesser $x$-coordinate first.

This leaves open only the possibility of a path which differs from one of the paths given in Figure \ref{fig:gadget} in its choice of whether to first visit $p_4$ or $p_5$.
For each of the paths of Figure \ref{fig:gadget}, it is easy to verify that swapping the positions of $p_4$ and $p_5$ in the path increases its weight.
\end{proof}

\subsection{Full construction}

We now describe the full construction, depicted in Figure \ref{fig:reduction}. Given a Max 2-SAT instance $\Pi$ = ($\{x_1, \cdots, x_n\}$, $\{c_1, \cdots, c_m\}$) we create a set of points $P_\Pi = S \cup C$. $S$ will be a special set containing exactly one point from each color class. $C$ is a collection of clause gadgets $C_1 \cup \cdots \cup C_m$. 
For each clause $c_i$, $C_i$ is anchored at $(i-1) (9\smallc + \mediumc)$.
Fix $W$ as the maximum $x$-coordinate of a point in $C$, i.e. $W = 27m\smallc + (m-1)\mediumc$.

We place the points of $S$ along the horizontal line $y = -(2W + 1)$ from East to West following a valid permutation. 
Points of the colors $T_i, F_i$ are coincident for each $i$, so that \emph{any} valid tour may visit $S$ via a straight line, while any invalid tour will have to backtrack.
The spacing around the points is proportional to the number of clauses that $x_i$ appears in.
The idea behind this spacing is to only penalize invalid tours relative to the number of clause gadgets they could ``cheat'' in (precisely those for which they are not $\alpha,\beta$-valid), since we cannot make the total spacing too large while preserving approximate solutions.
Specifically, for each variable $x_i$, let $n_i$ be the number of clauses that it or its negation appears in. Clearly $\sum_{i=1}^n n_i = 2m$, since each clause contains exactly two literals. Let $N_i = \sum_{j=1}^{i} n_j$. 
$S$ contains the following points for $(1 \leq i \leq n)$, where $\largec = \frac{27}{4} \smallc + (\frac{1}{4} - \frac{1}{m})\mediumc$:
\begin{itemize}
\item Point $r_{i}$ of color $R_i$ at ($W - 2\largec N_{i-1}$, $-2W - 1$)
\item Points $t_{i}, f_i$ of colors $T_i, F_i$ resp. both located at ($W - 2\largec N_{i-1} - \largec n_i$, $-2W - 1$)
\end{itemize}
Finally, $S$ contains a point $r_{n+1}$ of color $R_{n+1}$ at ($0$, $-2W - 1$) = ($W - 2\largec N_{n}$, $-2W - 1$).
We can now state the length of the optimal $\sigma$-tour for any valid $\sigma$.

\begin{lemma}
\label{lem:opttour}
Let $\sigma$ be a valid permutation satisfying $k$ clauses of $\Pi$.
Then there is a $\sigma$-tour $T^*$ of $P_\pi = S \cup C$ s.t.
\begin{align*}
    \|T^*\| =\ &5W + (m-1)\mediumc + k(c + 2a)\\
              &+ (m-k)(c + 2\sqrt{\smallc^2 + 4})
\end{align*}
where $c = f(a)$ is the constant defined in lemma \ref{lem:gadget}.
\end{lemma}
\begin{proof}
We construct a tour $T^*$ that first visits the points of $S$ via a straight line and then visits each of the clauses in ascending order $C_1, \cdots, C_m$, using the $\sigma$-path $T^*_i$ of Lemma \ref{lem:gadget} to visit each $C_i$.
First, write $S = \{s_1, \cdots, s_{3n+1}\}$ s.t. $s_1 \prec_\sigma \cdots \prec_\sigma s_{3n+1}$, and let $T_S = (s_1, \cdots, s_{3n+1})$ be the unique $\sigma$-path visiting $S$.
By construction, $||T_S|| = W$ for any valid $\sigma$.
Now, let $u_i, v_i$ be the first and last points of $T_i^*$, respectively.
Then $T^*$ consists of the edges in $T_S$, the edges in $T_1^*, \cdots, T_m^*$, and the additional edges $(s_{3n+1}, u_1), (v_m, s_{1}),$ and $(v_i, u_{i+1})$ for all $1 \leq i < m$.
By Lemma \ref{fig:gadget}, $\sum_{i \in [1,m]}T_i^* = k(c + 2\smallc) + (m-k)(c + 2\sqrt{\smallc^2 + 4})$.
Furthermore, $||s_{3n+1} - u_1|| = ||v_m, s_{1}||=2W$, and $\sum_{i=1}^m ||v_i - u_{i+1}|| = (m-1)\mediumc$.
Summing yields the desired total length.
\end{proof}

We now argue that the existence of a $\sigma$-cycle with length close to that of $T^*$ in Lemma \ref{lem:opttour} implies that $\sigma$ satisfies close to $k$ clauses of $\Pi$. 
To do so, we will need to break apart an arbitrary tour for a finer-grained analysis. 
With this goal in mind, let $T^\sigma$ be a $\sigma$-tour (for possibly invalid $\sigma$), and let $T^\sigma[S]$ and $T^\sigma[C]$ be the portions of $T^\sigma$ induced on $S$ and $C$ respectively. 
It is worth pointing out that, in general, $T^\sigma[S]$ and $T^\sigma[C]$ may not be connected. However, we have chosen a suitable separation of distance $2W$ between $S$ and $C$ so that any tour that alternates between $S$ and $C$ multiples times has length that well exceeds that of $T^*$. Thus, we can safely ignore these tours, and assume that both $T^\sigma[S]$ and $T^\sigma[C]$ are connected.

Now, let us further decompose $T^\sigma[C]$: 
we will use $T^\sigma_i$ to denote the portion of $T^\sigma[C]$ spent inside the clause gadget $C_i$.
More specifically, consider the set of vertical lines $L$ given by the equations $x = (27\smallc + \mediumc)(i-1)$ and $x = (27\smallc + \mediumc)(i-1) + 27\smallc$ for $i \in [1, m]$, and let us subdivide each leg of $T^\sigma[C]$ at each crossing of a line in $L$.
After subdivision, let $T^\sigma_i$ be the set of line segments of $T^\sigma[C]$ contained (inclusively) between the vertical lines $x = (27\smallc + \mediumc)(i-1)$ and $x = (27\smallc + \mediumc)(i-1) + 27\smallc$. ($T^\sigma_i$ may not be connected, and may include partial legs of $T^\sigma[C]$.)
Collectively, the $T^\sigma_i$'s do not account for the remaining portion of $T^\sigma[C]$ between gadgets, which has length at least $(m-1)b$.
We get the following bound:

\begin{obs}
\label{obs:tourlb}
$\|T^\sigma\| > \sum_{i=1}^m \|T^\sigma_i\| + 5W + (m-1)b$.
\end{obs}

If $\sigma$ is valid, we can use Lemma \ref{lem:gadget} to bound the number of clauses $\sigma$ satisfies based on the length of an approximate $\sigma$-tour.
Lemma \ref{lem:invalid} will account for invalid $\sigma$.

\begin{lemma}
\label{lem:solutionmap}
Let $T^\sigma$ be a $\sigma$-tour of $P_\Pi$ for valid $\sigma$. 
Let $c = f(a)$ be the constant given in Lemma \ref{lem:gadget}. 
If
    \[\|T^\sigma\| \leq 5W + (m-1)\mediumc + k(c + 2a) + (m-k)(c + 2\sqrt{\smallc^2 + 4}),\]
then $\sigma$ satisfies at least $k$ clauses of $\Pi$.
\end{lemma}
\begin{proof}
For the sake of contradiction, let us assume that $\sigma$ satisfies $k'$ clauses for some $k' < k$.
If every $T^\sigma_i$ is connected, then by Lemma \ref{lem:gadget},
\[\sum_{i=1}^m \|T^\sigma_i\| > k'(c + 2\smallc) + (m-k')(c + 2\sqrt{\smallc^2 + 4}).\]
For each disconnected $T^\sigma_i$, $T^\sigma$ incurs an additional length $b$ to enter and leave the gadget $C_i$.
We charge this overhead to $T^\sigma_i$, such that the above bound always holds as long as $b >> a$, e.g. $b > 30a$.
However, from Observation \ref{obs:tourlb} and the assumption of the lemma, we have that 
\[\sum_{i=1}^m \|T^\sigma_i\| < 2 + k(c + 2a) + (m-k)(c + 2\sqrt{\smallc^2 + 4}),\]
which implies that $(k -k')(2\sqrt{a^2+4} -2a) < 0$, a contradiction since $k - k' > 0$.
\end{proof}

We lastly rule out the possibility of a $\sigma$-tour for invalid $\sigma$.
Essentially, we are able to show that any gains made by ``cheating'' in the clause gadgets are offset by an equal increase in overhead to visit $S$.

\begin{lemma}
\label{lem:invalid}
Let $T^\sigma$ be a $\sigma$-tour of $P_\Pi$. There exists valid $\sigma'$ and $\sigma'$-tour $T^{\sigma'}$ such that $\|T^{\sigma'}\| \leq \|T^{\sigma}\|$. Furthermore, $\sigma'$ can be computed in polynomial time given $\sigma$.
\end{lemma}
\begin{proof}
If $\sigma$ is valid, then setting $\sigma' = \sigma$ satisfies the lemma.
Otherwise, consider the valid ordering $\sigma'$ that agrees with $\sigma$ on all variables, and the $\sigma'$-tour $T^*$ given by Lemma \ref{lem:opttour}.
We are done if $\|T^*\| \leq \|T^\sigma\|$; otherwise $\|T^\sigma\| < \|T^*\|$. 
Then, by Observation \ref{obs:tourlb}, $\sum_{i=1}^m \|T^\sigma_i\| < \|T^*\| - 5W - (m-1)b$.
We can similarly decompose $T^*$ to get $\|T^*\| = 5W + (m-1)b + \sum_{i=1}^m \|T^*_i\|$, where each $T^*_i$ is precisely the path visiting $C_i$ given by Lemma \ref{lem:gadget}.
Then $\sum_{i=1}^m \|T^\sigma_i\| <  \sum_{i=1}^m \|T^*_i\|$, and thus $\|T^\sigma_i\| < \|T^*_i\|$ for some $i$.
Therefore, let $\hat{T}^\sigma$ be the nonempty set of $T^\sigma_i$ for which $\|T^\sigma_i\| < \|T^*_i\|$.
We now reason that for each $T^\sigma_i \in \hat{T}^\sigma$, $T^\sigma$ incurs a penalty of $l > \|T^*_i\| - \|T^\sigma_i\|$ in $T^\sigma[S]$. 
In other words, if $\sum_{i=1}^m \|T^*_i\| - \sum_{i=1}^m \|T^\sigma_i\| \leq t$ for some $t$, then $\|T^\sigma[S]\| - \|T^*[S]\| \geq t$. 
This contradicts the earlier claim that $\|T^\sigma\| < \|T^*\|$, and will be sufficient to prove the lemma for $T^{\sigma'} = T^*$.

Consider some $T^\sigma_i$ corresponding to a clause $C_i$ on variables $x_{\alpha}, x_{\beta}$.
If $T^\sigma_i \in \hat{T}^\sigma$, then $\sigma$ is not $\alpha,\beta$-valid by Lemma \ref{lem:gadget}.
We now decompose $T^\sigma[S]$ by subdividing the tour at crossings of the lines $x = W - 2\largec N_{i-1}$ for $i \in [1, n+1]$. 
Let $\cost{i}$ denote the length of the tour between lines $x = W - 2\largec N_{i-1}$ and $x = W - 2\largec N_{i}$, i.e. $\cost{i}$ is the length of $T^\sigma[S]$  between the points of color $R_i$ and $R_{i+1}$.
Observe that $\cost{\alpha} = 2ln_\alpha$ and $\cost{\beta} = 2ln_\beta$ only if $\sigma$ is $\alpha,\beta$-valid. 
Otherwise, for either $\alpha$ or $\beta$, say without loss of generality $\alpha$, $\cost{\alpha} \geq 4ln_\alpha$.
Now, let $\hat{I} = \{i: \cost{i} \geq 4ln_i\}$.
Clearly, $\hat{I}$ is nonempty; else $\sigma$ would be valid.
Then $\|T^\sigma[S]\| - \|T^*[S]\| \geq 2l\sum_{i \in \hat{I}} n_i$.

Furthermore, $|\hat{T}^\sigma| \leq \sum_{i \in \hat{I}} n_i$, since $\sum_{i \in \hat{I}} n_i$ counts each clause $C_i$ on variables $x_\alpha, x_\beta$ for which $\sigma$ is not $\alpha,\beta$-valid at least once, and for all $i$, $T^\sigma_i \in \hat{T}^\sigma$ for clause $C_i$ on variables $x_\alpha, x_\beta$ implies $\sigma$ is not $\alpha,\beta$-valid.
We can trivially bound $\|T^\sigma_i\| > 0$, and since $\|T^*_i\| < 29\smallc + O(1)$, we get $\|T^*_i\| - \|T^\sigma_i\| < 30\smallc$ for sufficiently large choice of constant $a$. 
Then, $\sum_{i=1}^m \|T^*_i\| - \sum_{i=1}^m \|T^\sigma_i\| \leq 30\smallc|\hat{T}^\sigma| \leq 30\smallc\sum_{i \in \hat{I}} n_i$.
Thus, for $t = 2l\sum_{i \in \hat{I}} n_i$, it holds that 
$\|T^\sigma[S]\| - \|T^*[S]\| \geq t$ and $\sum_{i=1}^m \|T^*_i\| - \sum_{i=1}^m \|T^\sigma_i\| \leq t$, as long as $l \geq 30\smallc$. This completes the proof for $T^{\sigma'} = T^*$.
\end{proof}

We are ready to prove theorem \ref{thm:hardness}, which we restate here.
\hardness*
\begin{proof}
Assuming a PTAS for Euclidean PCTSP, we can derive a PTAS for Max 2-SAT, which contradicts its hardness \cite{papadimitriou1991}.
Given an instance to Max 2-SAT $\Pi$ with optimal solution value $k$ and target error $\epsilon$, we first construct the instance $P_\Pi$ of Figure \ref{fig:reduction} in polynomial time.
Let $f(k)$ be the length of the tour in Lemma \ref{lem:opttour}, namely $f(k) = 5W + (m-1)\mediumc + k(c + 2a) + (m-k)(c + 2\sqrt{\smallc^2 + 4})$.
(Recall that $W = 27ma + (m-1)b$, and $c$ is a constant in terms of $a$ and $b$.) 
For an appropriate $\epsilon'$, we apply the PTAS for EPCTSP to recover a $\sigma$-tour of $P_\Pi$ that has length at most $(1 + \epsilon')f(k)$, since OPT of $P_\Pi$ is at most $f(k)$ by Lemma \ref{lem:opttour}.
We can assume $\sigma$ is valid, else we can compute a valid $\sigma'$ that admits a $\sigma'$-tour with equal or better cost in polynomial time by Lemma \ref{lem:invalid}.
In the remainder of the proof, we show that $\sigma$ satisfies at least $(1 - \epsilon)k$ clauses of $\Pi$ for a sufficiently small choice of constant $\epsilon' = f(\epsilon, \smallc, \mediumc)$.

In particular, we have a $\sigma$-tour $T$ with $||T|| < (1 + \epsilon')f(k)$.
Now, we need to ensure that $||T|| < f((1 - \epsilon)k)$ to claim $\sigma$ satisfies $(1 - \epsilon)k$ clauses by Lemma \ref{lem:solutionmap}. 
For simplicity,
\begin{itemize}
    \item let constant $c_a = 2(\sqrt{a^2+4}-a) > 0$, and
    \item let constant $c_{a,b} = 5\cdot27a + 6b + c$,
\end{itemize}
such that $f(k) = c_{a,b}m - 6b - c_a k$.
Then we just need to choose $\epsilon'$ to satisfy
\[(1 + \epsilon')(c_{a,b}m - 6b - c_a k) \leq c_{a,b}m - 6b - c_a (1 - \epsilon) k.\]
We can rewrite this as
\[\epsilon' \leq \frac{\epsilon c_a k}{c_{a,b}m - 6b - c_a k}.\]
Finally, since it is trivial to satisfy half of the clauses in any 2-SAT instance, we can substitute $k \geq m/2$, and choose
\[\epsilon' = \frac{\epsilon c_a}{2c_{a,b} - c_a},\]
which is only a function of $\epsilon, a, b$ as claimed.
\end{proof}
\section{Conclusion}
In this paper, we introduced the Polychromatic TSP and studied its metric and Euclidean variants.
In the metric case, we gave a constant factor approximation that remains polynomial time for any number of colors.
We complemented this algorithmic result with the nonexistence of a PTAS, even for points in the plane as long as the number of colors is unbounded.
An interesting open question remains: does Euclidean PCTSP admit a PTAS for a constant number of colors?

\newpage

\bibliographystyle{plain}
\bibliography{citation}

\begin{thebibliography}{1}

\bibitem{anily1992}
Shoshana Anily and Refael Hassin.
\newblock The swapping problem.
\newblock {\em Networks}, 22(4):419--433, 1992.

\bibitem{arora1998}
S.~Arora.
\newblock Polynomial time approximation schemes for euclidean traveling salesman and other geometric problems.
\newblock {\em J. ACM}, 45(5):753–782, September 1998.

\bibitem{baligacs2024}
J.~Balig\'{a}cs, Y.~Disser, A.~E. Feldmann, and A.~Zych-Pawlewicz.
\newblock A $(5/3+\epsilon)$-approximation for tricolored non-crossing euclidean tsp.
\newblock In {\em 32nd Annual European Symposium on Algorithms (ESA 2024)}, volume 308, pages 15:1--15:15, 2024.

\bibitem{chalasani1996}
P.~Chalasani, R.~Motwani, and A.~Rao.
\newblock Algorithms for robot grasp and delivery.
\newblock In {\em 2nd International Workshop on Algorithmic Foundations of Robotics}, 1996.

\bibitem{dross2023}
F.~Dross, K.~Fleszar, K.~W{\k{e}}grzycki, and A.~Zych-Pawlewicz.
\newblock Gap-eth-tight approximation schemes for red-green-blue separation and bicolored noncrossing euclidean travelling salesman tours.
\newblock In {\em Proceedings of the 2023 Annual ACM-SIAM Symposium on Discrete Algorithms (SODA)}, pages 1433--1463, 2023.

\bibitem{karlin2021}
A.~R. Karlin, N.~Klein, and S.~O. Gharan.
\newblock A (slightly) improved approximation algorithm for metric tsp.
\newblock In {\em Proceedings of the 53rd Annual ACM SIGACT Symposium on Theory of Computing}, STOC 2021, page 32–45. Association for Computing Machinery, 2021.

\bibitem{papadimitriou1991}
C.~H. Papadimitriou and M.~Yannakakis.
\newblock Optimization, approximation, and complexity classes.
\newblock {\em Journal of Computer and System Sciences}, 43(3):425--440, 1991.

\bibitem{papadimitriou1993}
C.~H. Papadimitriou and M.~Yannakakis.
\newblock The traveling salesman problem with distances one and two.
\newblock {\em Mathematics of Operations Research}, 18(1):1--11, 1993.

\bibitem{srivastav2001}
A.~Srivastav, H.~Schroeter, and C.~Michel.
\newblock Approximation {Algorithms} for {Pick}-and-{Place} {Robots}.
\newblock {\em Annals of Operations Research}, 107(1):321--338, October 2001.

\end{thebibliography}

\end{document}